\documentclass[11pt]{article}
\usepackage[affil-it]{authblk}
\usepackage{amsfonts}
\usepackage{amsmath,amsthm}
\newtheorem{Theorem}{Theorem}
\newtheorem{Definition}[Theorem]{Definition}

\newtheorem{Proposition}[Theorem]{Proposition}
\newtheorem{Corollary}{Corollary}[Theorem]

\usepackage{fullpage}
\usepackage{stackrel}
\usepackage{amssymb}
\usepackage{graphicx}
\usepackage{tikz}
\usetikzlibrary{matrix,chains,positioning,decorations.pathreplacing,arrows}
\usepackage{qcircuit}
\usepackage{float}
\usepackage{braket}
\usepackage{subcaption}
\usepackage{mathrsfs}
\usepackage{cite}
\usepackage{physics}
 \usepackage{hyperref}
\hypersetup{
    colorlinks,
    citecolor=blue,
    filecolor=black,
    linkcolor=blue,
    urlcolor=blue
}
\usepackage{cleveref}
\newenvironment{proof1}{\textit{Proof Theorem 1:}}{\hfill$\square$}
\newenvironment{proof2}{\textit{Proof Theorem 2:}}{\hfill$\square$}
\newenvironment{proof3}{\textit{Proof Theorem 3:}}{\hfill$\square$}
\newenvironment{proof4}{\textit{Proof Theorem 4:}}{\hfill$\square$}
\newcommand{\spec}{\mathrm{spect}}
\newcommand{\Attr}{\mathrm{Attr}}
\newcommand{\Fix}{\mathrm{Fix}}
\newcommand{\Ker}{\mathrm{Ker}}
\newcommand{\RE}{\mathrm{Re}}
\newcommand{\supp}{\mathrm{supp}}

\begin{document}

\title{Number of steady states of quantum evolutions}

\author[1,2,*]{Daniele Amato}
\author[1,2]{Paolo Facchi}

\affil[1]{Dipartimento di Fisica, Università di Bari, I-70126 Bari, Italy}
\affil[2]{INFN, Sezione di Bari, Bari 70126, Italy}
\affil[*]{daniele.amato@uniba.it}

\date{\today}

\maketitle

\begin{abstract}
We prove sharp universal upper bounds on the number of steady and asymptotic states of discrete- and  continuous-time Markovian evolutions of 
open quantum systems. We show that the bounds  depend only on the dimension of the system and not on the details of the dynamics. A comparison with similar bounds deriving from a recent spectral conjecture for Markovian evolutions is also provided.  
\end{abstract}


\section{Introduction}

\label{intro}
Spectral theory is still a hot topic in quantum mechanics. Indeed, quantum theory was developed at the beginning of the last century in order to explain the energy spectra of atoms~\cite{History_QM}.

In particular, the dynamics of a closed quantum system, namely isolated from its surroundings, is encoded in the eigenvalues (energy levels) of its Hamiltonian~\cite{Teschl_Math}. Similarly, for an open quantum system under the Markovian approximation~\cite{breuer_petr}, studying the spectrum of the Gorini-Kossakowski-Lindblad-Sudarshan (GKLS) generator (the open-system analogue of the Hamiltonian) allows us to obtain information about the dynamics of the system~\cite{baum_narn_2008}. 

In spite of this general interplay between spectrum and dynamics, a complete understanding of open-quantum-system evolutions still remains a formidable task. However, a more detailed analysis may be performed if we restrict our attention to the large-time dynamics of open systems. This amounts to study the steady and, more generally, asymptotic states towards which the evolution converges in the asymptotic limit.

This  topic was already investigated at the dawn of the theory of open quantum systems in various works~\cite{Spohn_77,Frigerio_78,Frigerio_Verri_82} (see also the review~\cite{Spohn_rev}), 
focusing on the existence and the uniqueness of a steady state for Markovian evolutions. Moreover, the structure of steady and asymptotic manifolds were taken into account in several later articles~\cite{fagnola_2001,arias_2002,baum_narn_2008,Conserved_Albert,Wolf's_notes,Nigro_steady_states,AFK_asympt_1,AFK_asympt_2,AFK_asympt_3,Yoshida_uniq_steady}. 

 Besides their theoretical importance, stationary states also play a central role in reservoir engineering~\cite{Zoller_res_eng,Wolf_res_eng,Cirac_res_eng}, consisting of properly choosing the system-environment coupling for preparing a target quantum state, or in phase-locking and synchronization of quantum systems~\cite{jex_synchr}.
 
 Moreover, GKLS generators with multiple steady states~\cite{Albert_Thesis} may be used in order to drive a dissipative system into (degenerate) subspaces protected from noise~\cite{Zanardi_noi_sub} or decoherence~\cite{Lidar_dec_sub}, in which only a unitary evolution, related to purely imaginary eigenvalues of the generator~\cite{Conserved_Albert}, may be exploited for the realization of  quantum gates~\cite{Zan_2014,Zan_diss_prot_dyn,gen_ad_Paolo}. For this reason, the analysis of stationary states and, more generally, the study of the relaxation of an open quantum system towards the equilibrium is needed for applications in quantum information storage and processing~\cite{Noise_sub_2001,Dec_sub_2000,top_prot,top_prot_dis}.

 The asymptotic properties of open quantum systems have also been  deeply studied in quantum statistical mechanics. In particular, dissipative quantum phase transitions~\cite{QED_ph_tr,crit_prop}, as well as driven-dissipative systems~\cite{ph_tra,opt_cav}, require the study of the large-time dynamical behaviour of the system.
More generally, determining the steady states of an open system sheds light on the transport properties of the system itself. In particular, the existence of discontinuities of the dimension of the steady-state manifold should correspond to a  jump for the transport features of the system~\cite{benatti_XX_chain_1,benatti_XX_chain_2}. 
 
Finally, open-quantum-system asymptotics naturally emerges in quantum implementations of Hopfield-type attractor neural networks~\cite{hopfield1982neural}. Indeed, the stored memories of such type of network may be identified with the stationary states of its (non-unitary) evolution~\cite{sanpera_attractor_2021,sanpera_attractor_2022_corr}. 

Despite the much effort devoted to the asymptotic dynamics of open quantum systems, general constraints for the number of steady and asymptotic states of quantum evolutions are still to be found, as far as we know. Besides the theoretical relevance of this problem, they may allow us to elucidate the potential of some of the above mentioned applications.    

  In this Article, we find sharp \emph{universal} upper bounds on the number of linearly independent steady and asymptotic states of discrete-time and Markovian continuous-time quantum evolutions. Importantly, these bounds are only related to the dimension of the system and not on the properties of the dynamics.

  The Article is organized as follows. After introducing some preliminary notions in Section~\ref{prel}, we will discuss our main results in Section~\ref{Main_res}, then we will provide explicit examples proving the sharpness of the bounds in Section~\ref{sharp}. Subsequently, before proving the theorems in Section~\ref{proofs}, our results will be compared with analogous bounds derived from a recent universal spectral conjecture proposed in~\cite{Chru_CKK_1} in Section~\ref{rel_CKKS}. Finally, we will draw the conclusions of the work in Section~\ref{concl}.   


\section{Preliminaries}

\label{prel}
In the present Section we will recall some basic notions about evolutions of finite-dimensional open quantum systems, see also Section~\ref{proofs} for more details. 

The state of an arbitrary $d$-level open quantum system is given by a density operator $\rho$, namely a positive semidefinite operator on a Hilbert space $\mathcal{H}$ ($d=\dim\mathcal{H}$) with $\Tr\rho =1$, whereas its dynamics in a given time interval $[0,\tau]$ with $\tau>0$ is described by a quantum channel $\Phi$, namely a completely positive trace-preserving map (a superoperator) on $\mathcal{B}(\mathcal{H})$, the space of linear operators on $\mathcal{H}$~\cite{Hein_Ziman}. 

If the system state at time $t=0$ is $\rho$, its \emph{discrete-time evolution} at time $t=n\tau$, with $n\in \mathbb{N}$, will be given by the action of the $n$-fold composition $\Phi^n$ of the map $\Phi$, namely,
\begin{equation}
  \rho(n\tau)=\Phi^{n}(\rho), \qquad n=0,1,\dots.
  \label{eq:discev}
\end{equation}

As the Hilbert space $\mathcal{H}$ is finite-dimensional, $\mathcal{B}(\mathcal{H})$ is isomorphic to the space of complex matrices of order $d$. We will indicate the space of $d\times d^{\prime}$ matrices with complex entries by $\mathcal{M}_{d,d^\prime}(\mathbb{C})$ and, for the sake of simplicity, $\mathcal{M}_{d}(\mathbb{C})\equiv \mathcal{M}_{d,d}(\mathbb{C})$.

Let 
$\mu_{\alpha}$, $\alpha=0,\dots , n-1$, with $n\leqslant d^2$ be the distinct eigenvalues of $\Phi$, namely 
\begin{equation}
\Phi(A_{\alpha})=\mu_{\alpha}A_{\alpha}, 
\end{equation}
 with $A_{\alpha}$ being an eigenoperator corresponding to $\mu_\alpha$. The spectrum $\spec (\Phi)$ is the set of eigenvalues of $\Phi$.  Let $\ell_\alpha$ be the algebraic multiplicity~\cite{horn_john} of the eigenvalue $\mu_\alpha$, so that
$\sum_{\alpha=0}^{n-1} \ell_\alpha = d^2$.
 It is well known~\cite{Wolf's_notes} that: 
 
 \noindent i) the spectrum is contained in the unit disk, 
\begin{equation}
	\spec(\Phi)\subseteq\mathbb{D}, \qquad \mathbb{D}=\{ \lambda\in\mathbb{C} \,:\, |\lambda|\leqslant 1 \};
\end{equation}
 ii) 1 is always an eigenvalue, namely,
 \begin{equation}
  \mu_0=1\in \spec (\Phi);	
 \end{equation}
  iii) the spectrum is symmetric with respect to the real axis, i.e.,
  \begin{equation}
  	 \mu_\alpha\in \spec (\Phi) \; \Rightarrow \; \mu_\alpha^*\in \spec (\Phi), \qquad \text{and}\quad  \Phi(A^{\dagger}_\alpha)=\mu_{\alpha}^{\ast}A^{\dagger}_\alpha,
  \end{equation}
iv) the  unimodular or \textit{peripheral} eigenvalues $\mu_\alpha \in \partial\mathbb{D}$, the boundary of $\mathbb{D}$, are semisimple, i.e.\ their algebraic multiplicity $\ell_\alpha$ coincides with their geometric multiplicity.
  
  The eigenspace $\Fix(\Phi)$ corresponding to $\mu_{0}=1$, called the \textit{fixed-point space} of $\Phi$, is spanned by a set of $\ell_0$ density operators, which are the \textit{steady} (or stationary) states of the channel $\Phi$.
%
  
  Also, the space $\Attr(\Phi)$ corresponding to the peripheral eigenvalues $\mu_\alpha \in \partial\mathbb{D}$ is known as the asymptotic~\cite{albert_2019} or the \textit{attractor subspace}~\cite{jex_st_2012,jex_st_2018}   of the channel $\Phi$, since the evolution $\Phi^n(\rho)$ of any initial state $\rho$ asymptotically moves towards this space  for large times, i.e.\ as $n\to\infty$, see Section~\ref{proofs} for more details. These limiting states may be called oscillating or \textit{asymptotic states}, and it is always possible to construct a basis of such states for the subspace $\Attr(\Phi)$, analogously to $\Fix(\Phi)$. 
  
  Note that closed-system evolutions are described by a \textit{unitary channel} \begin{equation}
  \Phi(X)=UX U^{\dagger}, \quad \text{for all } X \in \mathcal{B}(\mathcal{H}),
\end{equation}
 and some unitary $U$. Importantly, a quantum channel is unitary if and only if $\spec (\Phi)\subseteq \partial\mathbb{D}$, i.e.\ all its eigenvalues belong to the unit circle~\cite{Wolf's_notes}. 

The \emph{Markovian continuous-time evolution} of an open quantum system is described by a quantum dynamical semigroup~\cite{Spohn_rev} 
\begin{equation}
  \rho(t)=\Phi_{t}(\rho)=e^{t\mathcal{L}}\rho, \qquad t\geqslant 0,
  \label{eq:contev}
\end{equation}
 where the generator $\mathcal{L}$ takes the well-known \textit{GKLS form}~\cite{GKS_76,Lindblad_76}
\begin{equation}
\label{GKLS}
\mathcal{L}(X)=-i[H, X]+\sum_{k=1}^{d^{2}-1}
 \left( A_{k} X A_{k}^{\dagger}-\frac{1}{2} \{ A_{k}^{\dagger}A_{k} , X \} \right) =\mathcal{L}_{H}(X)+\mathcal{L}_{D}(X) ,\quad  X \in \mathcal{B}(\mathcal{H}) ,
\end{equation}
where the square (curly) brackets represent the (anti)commutator, $H=H^\dag$ is the system Hamiltonian, the noise operators $A_{k}$ are arbitrary, 
and the first and second terms $\mathcal{L}_{H}$ and $\mathcal{L}_{D}$ in Eq.~\eqref{GKLS} are called the Hamiltonian and dissipative parts of the generator, respectively. Notice that the GKLS form~\eqref{GKLS} is not unique and, in particular, so is the decomposition of $\mathcal{L}$ into Hamiltonian and dissipative contributions. $\mathcal{L}$ is called a \textit{Hamiltonian generator} if $\mathcal{L}_D = 0$ for one (and hence all) GKLS representation~\eqref{GKLS}. 

If $\lambda_{\alpha}$, $\alpha=0,\dots,m-1$, with $m\leqslant d^2$, denote the distinct eigenvalues of $\mathcal{L}$, from the GKLS form one obtains that $\lambda_{0}=0$ and, given  an eigenoperator $X_{0}\geqslant 0$ corresponding to this eigenvalue, then $X_{0}/ \Tr(X_{0})$ is a steady state of $\Phi_{t}=e^{t\mathcal{L}}$~\cite{Alicki_Lendi}. The kernel of $\mathcal{L}$, i.e.\ the eigenspace corresponding to the zero eigenvalue, will be denoted by $\Ker(\mathcal{L})$. Moreover,  
\begin{equation}
  \lambda_{\alpha} \in \spec(\mathcal{L}) \Rightarrow \lambda_{\alpha}^{\ast} \in \spec(\mathcal{L}), \quad  \text{and} \quad \RE(\lambda_{\alpha})=-\Gamma_{\alpha}\leqslant 0
\end{equation}
with $\Gamma_{\alpha}$ being the  \textit{relaxation rates} of $\mathcal{L}$. These parameters, describing the relaxation properties of an open system~\cite{kimura_17}, may be experimentally measured. A condition for the relaxation rates of a quantum dynamical semigroup, recently conjectured in~\cite{Chru_CKK_1} and which we will call Chru{\'s}ci{\'n}ski-Kimura-Kossakowski-Shishido (CKKS) bound, is recalled in Section~\ref{rel_CKKS} in order to investigate its relation with the main results of this work, stated in Section~\ref{Main_res}. 

Finally, note that the purely imaginary (peripheral) eigenvalues of $\mathcal{L}$ are semisimple and are related to the large-time dynamics of $\Phi_t = e^{t\mathcal{L}}$, as the space corresponding to such eigenvalues is the asymptotic manifold $\Attr(\mathcal{L})$ of the Markovian evolution, see Section~\ref{proofs} for details. Importantly, as for unitary channels,  the generator $\mathcal{L}$ is Hamiltonian if and only if $\Gamma_\alpha = 0$ for all $\alpha = 0, \dots , m-1$, i.e.\ all its eigenvalues are peripheral.


\section{Bounds on the dimensions of the asymptotic manifolds}
\label{Main_res}
In this Section we will present the main results of this work, whose proofs are postponed to Section~\ref{proofs}. First, let us introduce the quantities involved in our findings.
Remember that we denoted with $\mu_{\alpha}$ the $\alpha$-th distinct eigenvalue of $\Phi$ and with $\ell_\alpha$ its algebraic multiplicity with 
$\alpha =0, \dots , n-1$. In particular,  $\ell_0$ is the algebraic multiplicity of $\mu_{0}=1$, and coincides with the dimension of its eigenspace, the steady-state manifold, i.e.\
 \begin{equation}
  \ell_0 = \dim \Fix(\Phi).
\end{equation}
We define the \textit{peripheral multiplicity} $\ell_{\mathrm{P}}$ of $\Phi$ as the sum of the multiplicities of all  peripheral eigenvalues, which coincides with the dimension of the attractor subspace $\Attr(\Phi)$, made up of asymptotic states. Namely,
\begin{equation}
\label{lP_def}
    \ell_{\mathrm{P}}= \sum_{\mu_{\alpha}\in \partial\mathbb{D}} \ell_{\alpha} = \dim\Attr(\Phi).
\end{equation}
\textit{Physically, $\ell_0$ and $\ell_P$ are respectively the number of  independent steady and asymptotic states of the evolution described by $\Phi$.}

Analogously, denote with $m_\alpha$ ($\alpha=0,\dots , m-1$) the algebraic multiplicity of the $\alpha$-th \emph{distinct} eigenvalue $\lambda_{\alpha}$ of the generator $\mathcal{L}$ of the continuous-time semigroup~\eqref{eq:contev}. In particular $m_0$ denotes the multiplicity of the zero eigenvalue $\lambda_{0}=0$, so that
\begin{equation}
  m_0 = \dim \Ker(\mathcal{L}).
\end{equation}
Moreover, the \textit{peripheral multiplicity} $m_{\mathrm{P}}$ of $\mathcal{L}$ is the sum of the multiplicities of its purely imaginary eigenvalues and measures the dimension of its attractor manifold:
\begin{equation}
\label{mP_def}
    m_{\mathrm{P}}= \sum_{\lambda_{\alpha} \in i\mathbb{R}}m_{\alpha} = \dim \Attr(\mathcal{L}).
\end{equation}

 \textit{The integers $m_0$ and $m_P$ represent respectively the number of  independent steady and asymptotic states of the Markovian evolution $\Phi_t = e^{t\mathcal{L}}$ generated by $\mathcal{L}$.}

Now we will provide sharp upper bounds on such multiplicities. 
Let us call a quantum channel \textit{non-trivial} if it is different from the identity channel, $\Phi(\rho)=\rho$.
\begin{Theorem}[Unitary discrete-time evolution]
Let $\Phi$ 
be a non-trivial \emph{unitary} quantum channel on a $d$-dimensional system. 
Then the  multiplicity $\ell_{0}$ of the eigenvalue 1 and the peripheral multiplicity $\ell_{\mathrm{P}}$ of $\Phi$ satisfy
\begin{equation}
\ell_{0} \leqslant d^{2}-2d+2,\qquad \ell_{\mathrm{P}} = d^2.
\label{sharp_m0_phi}
\end{equation}
\label{Th_1_qc}
\end{Theorem}

\begin{figure}[t!]
\begin{center}
\includegraphics[width=0.65\linewidth , height=4.4 cm]{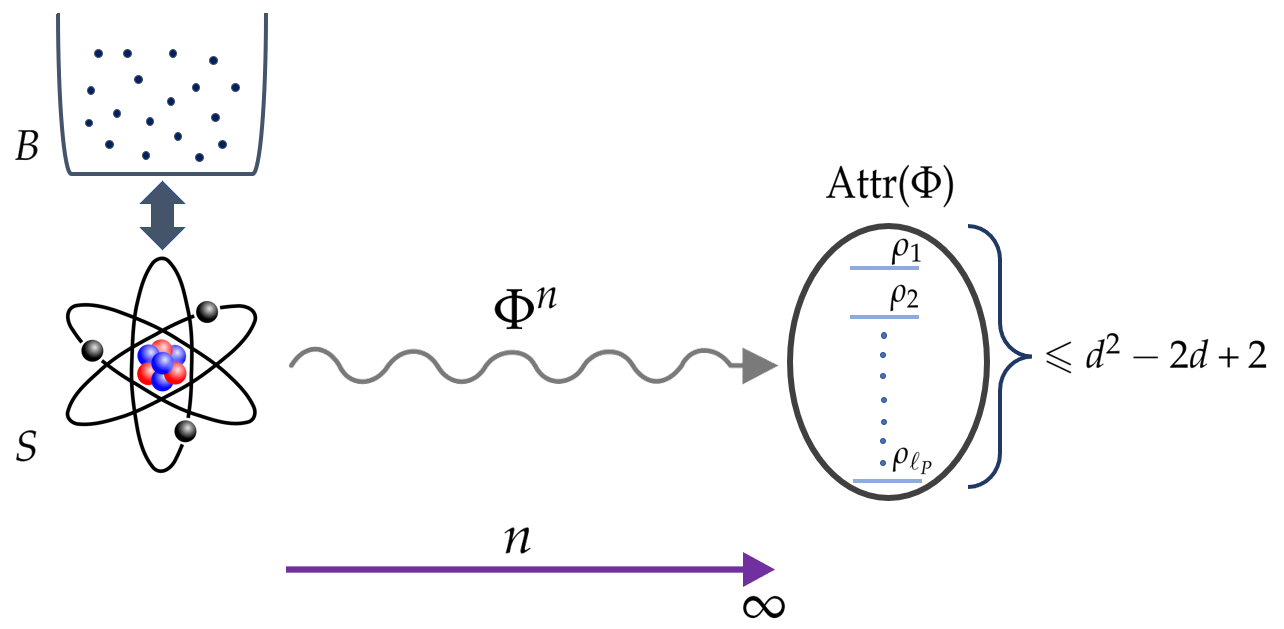}
\end{center}
\caption{Schematic representation of the content of Theorem~\ref{Th_2_qc}. A system $S$ coupled to a bath $B$ evolves according to the non-unitary discrete-time evolution $\Phi^n$ with $n\geqslant 1$. The asymptotic states $\rho_{1}\,,\dots , \rho_{\ell_{\mathrm{P}}}$ of $S$, spanning the attractor subspace $\Attr(\Phi)$, are at most $d^{2}-2d+2$, where $d$ is the dimension of the system.}
\label{figura1}
\end{figure}
 
\begin{Theorem}[Non-unitary discrete-time evolution]  Let $\Phi$ be a \emph{non-unitary} quantum channel.
Then the  multiplicity $\ell_{0}$ of the eigenvalue 1 and the peripheral multiplicity $\ell_{\mathrm{P}}$ of $\Phi$ obey
\begin{equation}
\ell_0 \leqslant \ell_{\mathrm{P}} \leqslant d^{2}-2d+2.
\label{sharp_mP_phi}
\end{equation}
\label{Th_2_qc}
\end{Theorem}
The content of the latter result is schematically illustrated in Fig.~\ref{figura1}.
Now, it is possible to construct quantum channels with $\ell_{0}$ and $\ell_{\mathrm{P}}$ attaining the equalities in Eqs.~\eqref{sharp_m0_phi} and~\eqref{sharp_mP_phi}, namely all the upper bounds are \textit{sharp}, see Section~\ref{sharp} for explicit examples. Obviously, for a trivial quantum channel $\ell_{0}=\ell_{\mathrm{P}} = d^{2}$, therefore the bounds~\eqref{sharp_m0_phi} and~\eqref{sharp_mP_phi} are not valid. 

The above results, valid for discrete-time evolutions~\eqref{eq:discev} are perfectly mirrored by the following results on Markovian continuous-time evolutions~\eqref{eq:contev}, with GKLS generators~\eqref{GKLS}.
\begin{Theorem}[Hamiltonian generator]
Let $\mathcal{L}$ be a non-zero \emph{Hamiltonian} GKLS generator. Then the  multiplicity $m_{0}$ of the zero eigenvalue and the peripheral multiplicity $m_{\mathrm{P}}$ of $\mathcal{L}$ fulfill
\begin{equation}
m_0 \leqslant d^2 - 2d+2,\qquad m_{\mathrm{P}} = d^2.
\label{sharp_m0_Ham_L}
\end{equation}
\label{Th_1_L}
\end{Theorem}
\begin{Theorem}[Non-Hamiltonian GKLS generator]  Let $\mathcal{L}$ be a \emph{non-Hamiltonian} GKLS generator. Then the  multiplicity $m_{0}$ of the zero eigenvalue and the peripheral multiplicity $m_{\mathrm{P}}$ of $\mathcal{L}$ satisfy
\begin{equation}
m_{0} \leqslant m_{\mathrm{P}} \leqslant  d^{2}-2d+2. 
\label{sharp_m0_L}
\end{equation}
\label{Th_2_L}
\end{Theorem}

The bounds~\eqref{sharp_m0_Ham_L} and~\eqref{sharp_m0_L} are also sharp as the previous ones, see Section~\ref{sharp}. Clearly, the two latter theorems do not apply to the zero operator because in such case $m_0 = m_{\mathrm{P}}=d^{2}$.

Theorem~\ref{Th_1_qc} provides a tight universal upper bound on the number of linearly independent steady states of a (non-trivial) unitary quantum channel $\Phi$, depending only on the dimension $d$ of the system. Similarly, Theorem~\ref{Th_2_qc} shows that the number of linearly independent steady and asymptotic states of a non-unitary channel $\Phi$ is bounded from above by the same $d$-dependent quantity. Theorems~\ref{Th_1_L} and~\ref{Th_2_L} provide analogous constraints for non-zero Hamiltonian and non-Hamiltonian generators respectively and, indeed, Theorem~\ref{Th_2_L} easily follows from Theorem~\ref{Th_2_qc}, as shown in Section~\ref{proofs}.

Interestingly, Theorem~\ref{Th_2_L} implies that when we add to a Hamiltonian generator a dissipative part, no matter how small, the peripheral multiplicity $m_{\mathrm{P}}$ undergoes a jump larger than the gap 
\begin{equation}
\Delta
=2(d-1),
\end{equation}
varying linearly with $d$. Consequently, we have ì\begin{equation}
N_{\mathrm{f}}
=2(d-1)-1
\end{equation}
forbidden values for $m_{\mathrm{P}}$.
 The same jump for the peripheral multiplicity $\ell_{\mathrm{P}}$ occurs when we pass from unitary channels to non-unitary ones according to the bound~\eqref{sharp_mP_phi}.
 
 \section{Sharpness of the bounds}
\label{sharp}
In this Section we will prove the sharpness of the bounds stated in Theorems~\ref{Th_1_qc}-\ref{Th_2_L}.
Let us start with the proof of the sharpness of the bound~\eqref{sharp_m0_Ham_L} for non-zero Hamiltonian GKLS generators. If we take 
\begin{equation}
\mathcal{L}(X) = -i[H , X], \qquad  H = h_1 \ketbra{e_1}{e_1} + h_2 \sum_{i=2}^{d} \ketbra{e_i}{e_i},
\qquad  h_1 \neq h_2 \in \mathbb{R},\quad X \in \mathcal{B}(\mathcal{H}),
\end{equation}
for some basis $\{ \ket{e_i} \}_{i=1}^d$ of $\mathcal{H}$, then it is immediate to check that \begin{equation}
  \Ker(\mathcal{L})=\operatorname{span}\{\ketbra{e_1}{e_1},  \ketbra{e_j}{e_k} \;:\; j,k=2,\dots,d\},
\end{equation}
whence $m_0 = (d-1)^2 + 1 = d^2 - 2d+2$. Furthermore, if we require that 
\begin{equation}
\label{eigen_con_Ham_gen}
h_1 - h_2 \neq 2k \pi,\;\; \forall k \in \mathbb{Z},
\end{equation}
the multiplicity $\ell_0$ of the corresponding unitary channel $\Phi = e^{\mathcal{L}}$ attains the inequality in Eq.~\eqref{sharp_m0_phi}. Note that condition~\eqref{eigen_con_Ham_gen} guarantees that $\Phi$ is not trivial. 

Let us now turn our attention to the sharpness of the bounds~\eqref{sharp_m0_L} for GKLS generators. 
Recall that the commutant $S^\prime$ of a system of operators $S = \{ A_{k} \}_{k=1}^{M} \subset \mathcal{B}(\mathcal{H})$ is defined as
\begin{equation}
S^{\prime}=\{ B \in\mathcal{B}(\mathcal{H}) \,:\, A_{k}B=BA_{k},\;\;k=1,\dots , M \}.
\end{equation}
Now consider the system $S=\{ A_{k}\}_{k=1}^{N}$ of diagonal operators with respect to the basis $\{ \ket{e_i} \}_{i=1}^d$ of $\mathcal{H}$ with
\begin{equation}
A_{k}=\lambda_{1}^{(k)} P_{1}+\lambda_{2}^{(k)} P_{2},\;\;k=1,\dots, N.
\label{diagonal}
\end{equation}
Here, $\lambda_{1}^{(k)},\lambda_{2}^{(k)} \in \mathbb{C}$, with $\lambda_{1}^{(k)}\neq\lambda_{2}^{(k)}$, are the eigenvalues of $A_k$, and
\begin{equation}
P_{1}= \ketbra{e_1}{e_1}
,\;\;P_{2}=\mathbb{I}-P_{1},
\end{equation}
are the corresponding spectral projections, with $\mathbb{I}$ being the identity operator on $\mathcal{H}$.
Note that, by construction, the eigenvalues $\lambda_{1}^{(k)},\lambda_{2}^{(k)}$ have respectively multiplicities $m_{1}=1$, $m_{2}=d-1$ for all $k=1,\dots, N$. Let us now take into account the generator
\begin{equation}
\mathcal{L}(X)=\sum_{k=1}^{N} A_{k} X A_{k}^{\dagger}-\frac{1}{2} \{ A_{k}^{\dagger}A_{k} , X \},
\label{sharp_L}
\end{equation}
for which $\mathcal{L}(\mathbb{I})= 0$. We have
\vspace*{0.2 cm}
\begin{equation}
m_{0}=\dim \Ker(\mathcal{L})=\dim S^{\prime}= \dim \{ P_{1} \}^{\prime}=d^{2}-2d+2 .
\end{equation}
Here, the second and fourth equalities follow respectively from Proposition~\ref{fixed_L} and Corollary~\ref{up_bound} in Section~\ref{proofs}, whereas
the third one is a consequence of Eq.~\eqref{diagonal}. Moreover, 
as $\mathcal{L}$ is non-Hamiltonian by construction, we necessarily have 
\begin{equation}
m_{0}=m_{\mathrm{P}}=d^{2}-2d+2,
\end{equation} 
by Theorem~\ref{Th_2_L}.
A quantum channel saturating the equalities in Eq.~\eqref{sharp_mP_phi} is simply $\Phi=e^{\mathcal{L}}$, with $\mathcal{L}$ given by Eq.~\eqref{sharp_L}.

In particular, we can construct a more physically transparent example of GKLS generator saturating the bounds~\eqref{sharp_m0_L} by taking $S=\{ P_{1} , P_{2} \}$. The associated Markovian channel $\Phi$ acts as follows with respect to the basis $\{ \ket{e_{i}} \}_{i=1}^d$
\begin{equation}
\mathcal{B}(\mathcal{H})\ni X=\begin{pmatrix}
X_{11} & X_{12} \\
X_{21} & X_{22} 
\end{pmatrix}\mapsto \Phi(X)=\begin{pmatrix}
X_{11} & X_{12}e^{-1} \\
X_{21}e^{-1} & X_{22} 
\end{pmatrix},
\end{equation}
with $X_{11} \in \mathbb{C}$, $X_{22}\in \mathcal{M}_{d-1}(\mathbb{C})$, $ X_{12}\in \mathcal{M}_{1,d-1}(\mathbb{C})$, and $X_{21}\in \mathcal{M}_{d-1,1}(\mathbb{C})$.

Therefore we realize that $\Phi$ is a  phase-damping channel causing an exponential suppression of the coherences $x_{12},\dots , x_{1d}\in X_{12}$, and we immediately see that it attains the equalities in the bounds~\eqref{sharp_mP_phi}, in line with the discussion above.

\section{Relation with the Chru{\'s}ci{\'n}ski-Kimura-Kossakowski-Shishido \\bound}
\label{rel_CKKS}
In this Section we will make a comparison between the bounds given in Theorems~\ref{Th_1_qc}-\ref{Th_2_L} and similar bounds arising from a recent spectral conjecture discussed in~\cite{Chru_CKK_1}.   
As already noted in Section~\ref{intro}, the real parts of the eigenvalues $\lambda_\alpha$ of a quantum dynamical semigroup $\Phi_t = e^{t\mathcal{L}}$ with GKLS generator $\mathcal{L}$ are non-positive.
However, it was recently conjectured in~\cite{Chru_CKK_1} that the relaxation rates $\Gamma_{\alpha} =-\Re(\lambda_\alpha)$ are not arbitrary non-negative numbers, but they must obey the CKKS bound
\begin{equation}
\Gamma_{\alpha} \leqslant \frac{1}{d} \sum_{\beta=1}^{m-1}m_\beta \Gamma_{\beta},\qquad \alpha=1,\dots, m-1,
\label{CKK-L}
\end{equation}
where $m_\beta$ is the algebraic multiplicity of $\lambda_\beta$.
This upper bound was not proved yet in general, but it holds for qubit systems, while for $d\geqslant 3$ it is valid for generators of unital semigroups, i.e.\ with $\mathcal{L}( \mathbb{I} )=0\,,$ 
  and for a class of generators obtained in the weak coupling limit~\cite{Chru_CKK_1} (see also~\cite{Chru_CKK_2} for further results). Also, it was experimentally demonstrated for two-level systems~\cite{Abragam,Slichter}.

The CKKS bound~\eqref{CKK-L} implies the following inequalities for non-Hamiltonian generators
\begin{equation}
m_{0}\leqslant m_{\mathrm{P}} \leqslant d^{2}-d.
\label{CKK-MP}
\end{equation}
Indeed, summing Eq.~\eqref{CKK-L} over the bulk, i.e. non-peripheral, eigenvalues of $\mathcal{L}$ yields
\begin{equation}
\sum_{\Gamma_\alpha<0} m_\alpha \Gamma_{\alpha} \leqslant \frac{m_{\mathrm{B}}}{d}\sum_{\Gamma_\beta<0}m_\beta \Gamma_{\beta},
\end{equation} 
where $m_{\mathrm{B}}=\sum_{\Gamma_\alpha<0}m_\alpha$ is the number of the repeated eigenvalues in the  bulk. If $\mathcal{L}$ is not Hamiltonian, viz. $m_B\neq0$ as noted in Section~\ref{prel}, this implies that
\begin{equation}
m_{\mathrm{B}} \geqslant d \Rightarrow m_{\mathrm{P}} \leqslant d^{2}-d,
\end{equation}
namely the assertion.

Interestingly, the CKKS bound~\eqref{CKK-L} implies also the following bound on the real parts $x_{\alpha}$ of the eigenvalues $\mu_{\alpha}$ of an arbitrary quantum channel~$\Phi$~\cite{Chru_CKK_1}
\begin{equation}
\sum_{\beta=0}^{n-1} \ell_\beta x_{\beta} \leqslant d(d-1)+dx_{\alpha},\qquad \alpha=1,\dots n-1,
\label{CKK-Phi}
\end{equation}
where $\ell_\beta$ is the algebraic multiplicity of $\mu_\beta$.

Although Eq.~\eqref{CKK-Phi} does not yield an upper bound similar to Eq.~\eqref{CKK-MP}  for the peripheral multiplicity $\ell_{\mathrm{P}}$ of $\Phi$, the multiplicity $\ell_{0}$ of the eigenvalue $\mu_{0}=1$, i.e.\ the number of steady states of $\Phi$, satisfies 
\begin{equation}
\ell_{0}\leqslant d^{2}-d,
\label{CKK_U_B}
\end{equation}
if $\Phi$ is not trivial.
The proof goes as follows: when $\ell_{0}=d^{2}$ we have the identity channel, so suppose $\ell_{0}=d^{2}-d+N$ with $0\leqslant N \leqslant d-1$. Then from Eq.~\eqref{CKK-Phi} one gets
\begin{equation}
x_{\alpha} \geqslant \frac{1}{d}\sum_{\beta=1}^{n-1}\ell_\beta x_{\beta}+\frac{N}{d},\qquad \alpha=1,\dots , n-1.
\label{average}
\end{equation}
 Now, the right-hand side of Eq.~\eqref{average} is the arithmetic mean of the  set 
\begin{equation}
S=\{ \underbrace{x_1 ,\dots , x_1}_{\ell_{1}},  \dots , \underbrace{x_{n-1} , \dots , x_{n-1}}_{\ell_{n-1}}, \underbrace{1, \dots ,1}_{N}  \}, 
\end{equation}
therefore condition~\eqref{average} is equivalent to require that all the elements of $S$ exceed their arithmetic mean, which is true if and only if $N=0$ and 
\begin{equation}
x_{1}=\dots = x_{n-1}=x\in [-1 , 1),
\end{equation} 
which concludes the proof of Eq.~\eqref{CKK_U_B}.

Furthermore, from Eq.~\eqref{CKK-MP} it follows that 
\begin{equation}
\ell_{0}\leqslant \ell_{\mathrm{P}}\leqslant d^{2}-d,
\label{Mar_CKK_b}
\end{equation}
for non-unitary Markovian channels, viz. of the form $\Phi=e^{\mathcal{L}}$ with $\mathcal{L}$ non-Hamiltonian generator. 

Now let us compare the bounds~\eqref{CKK-MP},~\eqref{CKK_U_B}, and~\eqref{Mar_CKK_b} arising from the CKKS conjecture~\eqref{CKK-L} discussed in the present Section with the ones stated in Section~\ref{Main_res}. First, the upper bound~\eqref{sharp_mP_phi} for $\ell_{\mathrm{P}}$ is also valid for non-Markovian channels, differently from the bound~\eqref{Mar_CKK_b} and, in the Markovian case, it is stricter than Eq.~\eqref{Mar_CKK_b} when $d \neq 2$. Analogously, the bound in Eq.~\eqref{sharp_m0_phi} and the one for $\ell_0$ in Eq.~\eqref{sharp_mP_phi} boil down to Eq.~\eqref{CKK_U_B} in the two-dimensional case, but they are stricter otherwise.    

Similarly, the bounds~\eqref{CKK-MP} for $m_{0}$ and $m_{\mathrm{P}}$ are not tight for all $ d\geqslant 3$, whereas they are equivalent to condition~\eqref{sharp_m0_L} in the case $d=2$. Consequently, the jump for $m_{\mathrm{P}}$ is also predicted by the bound~\eqref{CKK-MP}  but $\Delta=d$, which is loose for all $d\neq 2$. 
In  conclusion, the bounds given in Theorems~\ref{Th_1_qc}-~\ref{Th_2_L} imply the bounds~\eqref{CKK-MP},~\eqref{CKK_U_B}, and~\eqref{Mar_CKK_b} deriving from the CKKS conjecture~\eqref{CKK-L}, in favor of the validity of the conjecture itself.

\section{Proofs of Theorems~\ref{Th_1_qc}-\ref{Th_2_L}}
\label{proofs}
In this Section we will prove Theorems~\ref{Th_1_qc}-\ref{Th_2_L} stated in Section~\ref{Main_res}.  
To this purpose, let us recall several preliminary concepts, besides the ones introduced in Section~\ref{prel}. 

First, given a quantum channel $\Phi$, it always admits a Kraus representation~\cite{Hein_Ziman},
\begin{equation}
\Phi(X)=\sum_{k=1}^{N} B_{k}X B_{k}^{\dagger},
\qquad \sum_{k=1}^{N} B_{k}^{\dagger}B_{k}=\mathbb{I}, \qquad \text{with } X \in \mathcal{B}(\mathcal{H}),
\label{tr-pr}
\end{equation}
in terms of some operators $\{ B_{k} \}_{k=1}^N \subset \mathcal{B}(\mathcal{H})$. Note that the second equation in~\eqref{tr-pr} expresses the trace-preservation condition.

Let $\mu_{\alpha}\,,\alpha=0, \dots , n-1$, with $\mu_{0}=1$ ($\lambda_{\alpha}\,,\alpha=0,\dots , m-1$, with $\lambda_{0}=0$) denote the $n$ ($m$) distinct eigenvalues of $\Phi$ ($\mathcal{L}$).
Let $\mathcal{L}_{\alpha}$ ($\mathcal{M}_{\alpha}$) be the algebraic eigenspace of $\Phi$ ($\mathcal{L}$) corresponding to the eigenvalue $\mu_{\alpha}$ ($\lambda_{\alpha}$), whose dimension is the algebraic multiplicity $\ell_{\alpha}$ ($m_{\alpha}$) of the eigenvalue. 
The attractor subspaces of $\Phi$ and $\mathcal{L}$ read 
\begin{align}
&\Attr(\Phi)=\bigoplus_{\mu_{\alpha}\in \partial\mathbb{D} }\mathcal{L}_{\alpha},\\
&\Attr(\mathcal{L})=\bigoplus_{\lambda_{\alpha} \in i\mathbb{R}}\mathcal{M}_{\alpha},
\end{align}
whose dimensions are the peripheral multiplicities $\ell_P$ and $m_P$ of $\Phi$ and $\mathcal{L}$ defined in Eqs.~\eqref{lP_def} and~\eqref{mP_def} respectively.
Let $\Fix(\Phi)$ stand for the fixed-point space of $\Phi$, i.e.
\begin{equation}
\Fix(\Phi)= \{ A\in\mathcal{B}(\mathcal{H}) \,:\, \Phi(A)=A \},
\end{equation}
and $\Fix(\Phi^{\ast})$ indicate the fixed-point space of the dual $\Phi^{\ast}$ of $\Phi$, defined via 
\begin{equation}
\Tr(A\Phi(B))=\Tr(\Phi^{\ast}(A)B),\qquad A\,,\,B\in\mathcal{B}(\mathcal{H}).
\end{equation}
Note that $\Phi^\ast$ has the same eigenvalues with the same algebraic multiplicities of $\Phi$~\cite{Kato_per}.
In addition, the spectral projections $\mathcal{P}$ and $\mathcal{P}_{\mathrm{P}}$ onto $\Fix(\Phi)$ and $\Attr(\Phi)$ are quantum channels themselves~\cite{Wolf's_notes}.
Finally, let  $\mathcal{M}_{0}\equiv \Ker(\mathcal{L})$ be the kernel of $\mathcal{L}$, given by
\begin{equation}
\Ker(\mathcal{L})=\{ A\in\mathcal{B}(\mathcal{H}) \,:\, \mathcal{L}(A)=0 \}.
\end{equation}
Before discussing the proofs of Theorems~\ref{Th_1_qc}-\ref{Th_2_L}, we need a few preparatory results.

Consider $A\in\mathcal{B}(\mathcal{H})$ with spectrum $\spec(A)=\{ \lambda_{k}\}_{k=1}^{N}$. If $m_{k}\,,n_{k}$ are the algebraic and geometric multiplicities of the eigenvalue $\lambda_{k}$, let $d_{j,k}$ with $j=1,\dots , n_{k}$ and $k=1,\dots , N$ indicate the dimension of the $j$-th Jordan block corresponding to the eigenvalue $\lambda_{k}$ of the Jordan normal form $J$ of $A$~\cite{horn_john}.
\begin{Proposition}[\!\!\cite{conj_groups}]
Let $A\in\mathcal{B}(\mathcal{H})$ with Jordan normal form $J \in \mathcal{M}_{d}(\mathcal{C})$. Then
\begin{equation}
c_{A} =\dim \{ A \}^{\prime}=\sum_{k=1}^{N}\sum_{i=1}^{m_{k}}s_{i,k}^{2},
\end{equation}
where
\begin{equation}
s_{i,k}=| \{ j=1,\dots , n_{k} \,:\, d_{j,k} \geqslant i \} |,
\end{equation}
with $i=1,\dots , m_{k}$, $k=1,\dots , N$ and $|I|$ being the cardinality of the set $I$.
\end{Proposition}
\begin{Corollary}
Let $A\in\mathcal{B}(\mathcal{H})$ with $A \neq c\mathbb{I},\;c\in\mathbb{C}$. Then
\begin{equation}
c_{A} \leqslant d^{2}-2d+2.
\end{equation}
In particular, the equality holds if and only if $A$ is diagonalizable with spectrum $\spec(A)=\{ \lambda_{1} , \lambda_{2} \}$ having algebraic multiplicities $m_{1}=1$ and $m_{2}=d-1$.
\label{up_bound}
\end{Corollary}
\begin{proof} By definition $\{ s_{i,k} \}_{i=1}^{m_{k}}$ is a partition of $m_{k}$, so
\begin{equation}
c_{A} \leqslant \sum_{k=1}^{N} m_{k}^{2},
\end{equation} 
where the equality holds if and only if $A$ is diagonalizable, viz. $m_{k}=n_{k}$ for all $k=1,\dots , N$.
Now, by the fundamental theorem of algebra~\cite{horn_john}, $\{ m_{k}\}_{k=1}^{d}$ is a partition of $d$, consequently
\begin{equation}
c_{A} \leqslant d^{2},
\end{equation}
where the equality holds if and only if $A=c\mathbb{I}\,,c\in\mathbb{C}$. If $A$ is a non-scalar matrix, then the maximum value is attained when $A$ is diagonalizable and has spectrum $\spec(A)=\{ \lambda_{1} , \lambda_{2} \}$ with multiplicities $m_{1}=1$ and $m_{2}=d-1$, and it reads
\begin{equation}
c_{A}^{\mathrm{(max)}}=m_{1}^{2}+m_{2}^{2}=d^{2}-2d+2,
\end{equation}
which concludes the proof. 
\end{proof}

Let us now recall several known facts about open-system asymptotics. Let us start with the following definition.
\begin{Definition}[\!\!\cite{albert_2019}]
\label{faith_def}
A quantum channel $\Phi$ is said to be faithful if it admits an invertible steady state, i.e.\ $\Phi(\rho)=\rho >0$ invertible state.
\end{Definition}
The structure of the fixed-point space of the dual of a quantum channel is related to its Kraus operators in the following way.
\begin{Proposition}[\!\!\cite{arias_2002}]
Let $\Phi$ be a quantum channel with Kraus operators $\mathcal{B}=\{ B_{k} , B_{k}^{\dagger} \}_{k=1}^{N}$. Then
\begin{equation}
\mathcal{B}^{\prime} \subseteq \Fix(\Phi^\ast).
\label{fix_point_Phi}
\end{equation}
Furthermore, if $\Phi$ is faithful, then the equality holds in Eq.~\eqref{fix_point_Phi}.
\label{fixed_p}
\end{Proposition}
Let us now state the analogue of the latter result for GKLS generators, exploiting the representation~\eqref{GKLS}.
\begin{Proposition}[\!\!\cite{Wolf's_notes}]
Let $\mathcal{L}$ be a GKLS generator of the form~\eqref{GKLS} with $\mathcal{A}=\{ H , A_{k} , A_{k}^{\dagger} \}_{k=1}^{d^{2}-1}$. Then
\begin{equation}
\mathcal{A}^{\prime}\subseteq\Ker(\mathcal{L}^{\ast}),
\end{equation}
and the equality is satisfied if there exists an invertible state $0\!<\! \rho \in \Ker(\mathcal{L})$.
\label{fixed_L}
\end{Proposition}
Finally, the following proposition shows that we can reduce to faithful channels for the analysis of the fixed-point space.  
\begin{Proposition}[\!\!\cite{Wolf's_notes}]
Given a quantum channel $\Phi$, define the map $\varphi_{00}$ as
\begin{equation}
\Phi(X) =\varphi_{00}(X_0) \oplus 0 ,\qquad X = X_0 \oplus 0\in \mathcal{B}(\mathcal{H}_0) \oplus 0,
\label{def_Phit}
\end{equation}
where $\mathcal{H}_0:=\supp(\mathcal{P}(\mathbb{I}))$, i.e. the support space~\cite{Teschl_Math} of $\mathcal{P}(\mathbb{I})$. Then $\varphi_{00}$ is a faithful quantum channel and
\begin{equation}
\Fix(\Phi)=\Fix(\varphi_{00})\oplus 0,
\end{equation}
with $\Fix(\varphi_{00})$ indicating the fixed-point space of $\varphi_{00}$.
\label{Phi_Phitilde}
\end{Proposition}
Now we are ready to prove Theorems~\ref{Th_1_qc}-\ref{Th_2_L}.

\begin{proof1} Let $\Phi\neq\mathsf{1}$ be a unitary quantum channel with unitary $U$. Then it is easy to see from the spectral decomposition of $U$ that
\begin{equation}
\label{eigen_unit_ch}
\mu_{k\ell} = \lambda_k \lambda_\ell^\ast,\;\; k,\ell = 1, \dots , d,
\end{equation}
where $\lambda_k$, with  $k=1,\dots,d$ are the (repeated) unimodular eigenvalues of $U$. Thus the maximum value of the algebraic multiplicity $\ell_0$ of the eigenvalue 1 of $\Phi$ is $d^2 - 2d + 2$, achieved when $\lambda_1 = \dots = \lambda_{d-1} \neq \lambda_d$.

The equality $\ell_{\mathrm{P}} = d^2$ is trivial because all the eigenvalues of a unitary channel are peripheral, as it is clear from Eq.~\eqref{eigen_unit_ch}.
\end{proof1}
 
 \vspace{1mm}
\begin{proof2}
 First, let us prove that $\ell_0 \leqslant d^2 - 2d + 2$ for any non-unitary channel. 
 Let $\Phi$ be a non-unitary channel with Kraus operators $\mathcal{B}=\{ B_{k} , B_{k}^{\dagger} \}_{k=1}^{N}$. In the faithful case, see Definition~\ref{faith_def}, by applying Corollary~\ref{up_bound} and Proposition~\ref{fixed_p}, we obtain
\begin{equation}
\ell_{0}=\dim \Fix(\Phi^\ast) =\dim \mathcal{B}^{\prime}\leqslant d^{2}-2d+2.
\end{equation}
If $\Phi$ is not faithful, then we can define the faithful channel $\varphi_{00}$ as in Eq.~\eqref{def_Phit}, therefore we have as a consequence of Proposition~\ref{Phi_Phitilde}
\begin{equation}
\label{ub_l0}
\ell_{0}=\dim \Fix(\varphi_{00}^\ast) =\dim \mathcal{B}_0^{\prime}\leqslant d_0^{2} <  d^{2}-2d+2.
\end{equation}
where $d_0=\dim(\mathcal{H}_0) \leqslant d-1$ and $\mathcal{B}_0=\{ B_{0,k} , B_{0,k}^{\dagger} \}_{k=1}^{N_0}$ is the system of Kraus operators of $\varphi_{00}$. 

 Let us now prove the analogous bound on the peripheral multiplicity $\ell_{\mathrm{P}}$ of $\Phi$. Observe that the spectral projection $\mathcal{P}_{\mathrm{P}}$ of $\Phi$ onto $\Attr(\Phi)$ satisfies
\begin{equation}
\mathcal{P}_{\mathrm{P}}\neq \mathsf{1},
\end{equation}
because not all the eigenvalues of the non-unitary channel $\Phi$ are peripheral.
Indeed, $\mathcal{P}_{\mathrm{P}}$ is a non-unitary channel as $\mathcal{P}_{\mathrm{P}}$ is non-invertible.
Therefore, since the fixed-point space of $\mathcal{P}_{\mathrm{P}}$ is $\Attr(\Phi)$, it is sufficient to apply the bound~\eqref{ub_l0} to $\mathcal{P}_{\mathrm{P}}$. 
\end{proof2}

 \vspace{1mm}
\begin{proof3}
Let $\mathcal{L}$ be a non-zero Hamiltonian GKLS generator. Then it is straightforward to show that~\cite{gen_ad_Paolo}
\begin{equation}
\label{eigen_Ham}
\lambda_{k \ell} = -i(h_k - h_\ell),\qquad k,\ell = 1, \dots , d,
\end{equation}
where $h_k$, with $k=1,\dots,d$, are the (repeated) real eigenvalues of the Hamiltonian $H$. Therefore this implies that the maximum value of the algebraic multiplicity $m_0$ of the zero eigenvalue of $\mathcal{L}$ is $d^2-2d+2$, obtained by setting $h_1 = h_2 = \dots = h_{d-1} \neq h_{d}$. The equality $m_{\mathrm{P}} = d^2$  follows immediately from Eq.~\eqref{eigen_Ham}.  
\end{proof3} 

 \vspace{1mm}
\begin{proof4}
Let $\mathcal{L}$ be a non-Hamiltonian GKLS generator. The first inequality is trivial. Since
\begin{equation}
m_{\mathrm{P}}=\dim \Attr(\Phi)=\dim \Attr(\mathcal{L}),
\end{equation}
where $\Attr(\Phi)$ is the attractor subspace of the non-unitary channel $\Phi=e^{\mathcal{L}}$, the second inequality follows from Theorem~\ref{Th_2_qc}.   
\end{proof4} 

Notice that the universal bounds given in Theorems~\ref{Th_1_qc}-\ref{Th_2_L} may also be proved by using the structure theorems on the asymptotic evolution of quantum channels~\cite{Wolf's_notes,AFK_asympt_1}.

\section{Conclusions and outlooks}
\label{concl}
We found dimension-dependent sharp upper bounds on the number of independent steady states of non-trivial unitary quantum channels and an analogous bound on the number of  independent  asymptotic states of non-unitary channels. 
Moreover, similar sharp upper bounds on the number of independent steady and asymptotic states of  GKLS generators were also obtained. We further made a comparison of our bounds with similar ones obtained from the CKKS conjecture~\eqref{CKK-L} and~\eqref{CKK-Phi}.

Interestingly, the upper bound on the peripheral multiplicity of GKLS generators reveals that adding a dissipative perturbation to an initially Hamiltonian generator causes a jump for the peripheral multiplicity across a gap linearly depending  on the dimension, and an analogous remark may be made for the peripheral multiplicity $\ell_{\mathrm{P}}$ of quantum channels on the basis of condition~\eqref{sharp_mP_phi}.

These findings may be framed in a series of works, addressing the general spectral properties of open quantum systems~\cite{Chru_CKK_1,Chru_CKK_2,spectra_ph_tr,spectra_rand_2019,zyn_random_2021,spectra_rand_2021, Wolf's_notes}, in particular Markovian ones, and can motivate further study of the spectral properties of channels and generators, far from being completely understood. In particular, the bounds found in this Article may be the  consequence of a generalization of the CKKS bound~\eqref{CKK-L} involving also the imaginary parts of the eigenvalues of a GKLS generator. 
Moreover, structure theorems on the asymptotic dynamics~\cite{AFK_asympt_1,Wolf's_notes} may be employed in order to find further constraints for the quantities studied in this work.

\section*{Acknowledgements}
This work was partially supported by Istituto Nazionale di Fisica Nucleare (INFN) through the project ``QUANTUM'', by PNRR MUR Project PE0000023-NQSTI, by Regione Puglia and QuantERA ERA-NET Cofund in Quantum Technologies (Grant No. 731473), project PACE-IN, by the Italian National Group of Mathematical Physics (GNFM-INdAM), and by the Italian funding within the ``Budget MUR - Dipartimenti di Eccellenza 2023--2027'' (Law 232, 11 December 2016) - Quantum Sensing and Modelling for One-Health (QuaSiModO).
\bibliography{dim_biblio}
\bibliographystyle{unsrt}

\end{document}